\documentclass{amsart}

\usepackage{amsmath,amssymb,amsfonts,amsthm,mathtools}
\usepackage{mathrsfs}
\usepackage{bm}
\usepackage{cancel}
\usepackage{xfrac}

\usepackage{physics}

\usepackage{tikz-cd}

\usepackage[a4paper,margin=1.2in]{geometry}
\usepackage{hyperref}

\usepackage{amssymb}

\numberwithin{equation}{section}

\newtheorem{theorem}{Theorem}[section]
\newtheorem{lemma}[theorem]{Lemma}
\newtheorem{proposition}[theorem]{Proposition}
\newtheorem{corollary}[theorem]{Corollary}
\newtheorem{remark}[theorem]{Remark}

\theoremstyle{definition}
\newtheorem{definition}[theorem]{Definition}

\theoremstyle{plain}
\usepackage[numbers]{natbib}

\newtheoremstyle{solutionstyle}
  {}      
  {}      
  {\itshape}
  {}
  {\bfseries}
  {.}
  { }
  {}

\theoremstyle{solutionstyle}

\title{Quantum fidelity on Krein and $S$-spaces}

\author[]{%
Morgan Jones\\
\small Department of Mathematics\\
\small University of Oklahoma\\
\small Norman, Oklahoma, USA\\
\small \texttt{morgan.d.jones-2@ou.edu}
}
\email{morgan.d.jones-2@ou.edu}

\AtBeginDocument{\RenewCommandCopy\qty\SI}
\begin{document}

\maketitle

\enlargethispage{3\baselineskip}
\tableofcontents
\section{Abstract}

The notion of Fidelity for quantum states is a measure of how much two states overlap. In the matrix formalism of quantum mechanics, states are represented by density operators i.e. positive semi-definite matrices with trace equal to $1$ in a complex Euclidean space $M_n(\mathbb{C})$. In \cite{FelipeSosa2022}, the notion of quantum states on certain Krein spaces with indefinite metric induced by a fundamental symmetry $J$ is introduced, these are given the name $J$-states. We will define an analogous notion of measurement for $J$-states to the regular quantum theory and use it to show that a notion of fidelity holds in the Krein setting. We will also show that there exists an analogous result to the Fuchs-Caves measurement holds in this setting. We will then, following the developments in \cite{bag2024quantumuchannelssspaces}, extend this definition of fidelity to $U$-quantum states on $S$-spaces. We will demonstrate that the analogous geometric motivation holds in the Krein and $S$-space setting, as holds for quantum fidelity and geometric means of operators.

\section{Introduction}
Before we consider the notions of Krein and $S$-space fidelity, we recall the basics of quantum fidelity. We will keep this in line, as in the rest of this paper, with the spirit of \cite{Watrous_2018} and in this introduction we  will recall various useful facts from Chapters 1,2 and 3 from this reference. An alphabet $\Sigma$ is a finite set of $n$ elements. The space $\mathbb{C}^\Sigma$ consists of complex functions $u:\Sigma\to \mathbb{C}$. For example if $\Sigma=\{1,\dots,n\}$, then each $u\in\mathbb{C}^{\Sigma}$ satisfies \begin{equation*}
    u=(u(1), u(2),\dots,u(n))
\end{equation*}
This can then be endowed with the standard Euclidean inner-product for such vectors $u$ and $v$ \begin{equation*}\langle u,v\rangle u\cdot v=\sum_{a\in \Sigma}\overline{u(a)}v(a)\end{equation*}
We note of course that this has the different convention to abstract functional analysis where the complex conjugate is taken in the second factor.\\
A register $X$ is the abstraction of the idea of a set of finite values that some system could take. This could be in the form of an alphabet $\Sigma$, or the cartesian product of several smaller registers. If $X$ is a register represented by the alphabet $\Sigma$ then $\mathbb{C}^\Sigma$ is the complex Euclidean space associated to this register. If the register is a compound register $X=(Y_1,\dots,Y_n)$, then this complex Euclidean space is
\begin{equation*}
    \mathcal{X}=\mathcal{Y}_1\otimes \cdots\otimes \mathcal{Y}_n
\end{equation*}
Given a complex Euclidean space $\mathcal{X}$, a quantum state is represented by a positive semidefinite linear operator from $\mathcal{X}$ to itself with trace equal to $1$. Such matrices are called density operators. Positive semidefinite operators form a cone $\mathrm{Pos}(\mathcal{X})$, in the vector space of all linear operators on the space $\mathcal{X}$. In physical systems, one obtains information about quantum states via ``measurements". The quantum information formalism of a measurement is a function \begin{equation*}
    \mu: \Sigma \to \mathrm{Pos}(\mathcal{X}) 
\end{equation*}
 such that $\sum_{a\in\Sigma}\mu(a)=\mathbf{1}_{\mathcal{X}}$. The formalism is that if a system is in the state $P$ just before measurement, the probability of picking a given element $a$ of a register $\Sigma$ is given by $\langle \mu(a), P\rangle$, wher $\langle\cdot,\cdot\rangle$ is the Hilbert-Schmidt inner-product. Given two positive semidefinite operators $P$ and $Q$, the quantity known as fidelity is defined as 
\begin{equation*}
    F(P,Q)=\left\|\sqrt{P}\sqrt{Q}\right\|_1=\mathrm{Tr}\left(\sqrt{\sqrt{Q}P\sqrt{Q}}\right)
\end{equation*} where the Schatten $1$-norm is taken. This definition is actually the minimization of the quantity 
\begin{equation*}
    \sum_{a\in\Sigma}\langle \sqrt{\mu(a),P\rangle}\sqrt{\langle \mu(a),Q\rangle} 
\end{equation*}
over all possible measurements $\mu$. This geometrically is the minimization over all possible measurements $\mu$ of the marginal distribution of probability vectors $p$ and $q$
\begin{equation*}
    \sum_{a\in \Sigma}\sqrt{p_aq_a}
\end{equation*}
where the $a$-th entries of $p$ and $q$ are $\langle \mu(a),P\rangle$ and $\langle \mu(a),Q\rangle$ respectively.
The question we will address is what happens when we equip the space of complex linear operators with an indefinite inner product. The first indefinite setting we will consider is that of the Krein space. A Krein space is a space $\mathcal{H}$ equipped with an indefinite nondegenerate inner-product $[ \cdot,\cdot]$ under which the space can be decomposed into positive and negative parts where elements have either positive or negative inner-product with themselves. A Krein inner-product can always be written using the ambient inner-product of the Hilbert space as $\langle J\cdot,\cdot \rangle$ where $J$ is a self-adjoint involution in the space of linear operators of the space $\mathcal{H}$. After we have developed this theory, we will then move into the setting of $S$-spaces, where the inner-product is still indefinite but we do not require the self-adjointness of the matrix inducing the indefinite inner-product. Several authors (\cite{FelipeSosa2022}, \cite{bag2024quantumuchannelssspaces}, \cite{Heo2022}) have already started to consider the appropriate definitions for objects such as states, channels and measurements in these indefinite settings , where  The approaches in these two matrices are very much motivated by the theory of $C^*$-algebras. In this paper, we will avoid this setting. We will formalize the coming notions from a more geometric standpoint, using and adapting the work \cite{franco2026conejhermitianmatricesgeometric}. Our constant aim is to show that the definitions of fidelity we construct behave in analogous ways to the regular quantum theory. The structure we will approach this with is as follows: 
\begin{itemize}
    \item In Section \ref{Section 3} we will introduce the notational conventions that will be used to consider the quantum objects under consideration. This is especially important as we will use notation in line with \cite{Watrous_2018} which is different to many other recent papers in this area \cite{bag2024quantumuchannelssspaces}, \cite{FelipeSosa2022}, \cite{Heo2022})
    \item In Section \ref{Section 4}, we will review some of the differential geometric notions of $J$-spaces which will form the philosophical backbone of our definition of $J$-fidelity and later of $U$-fidelity. These facts are largely taken from \cite{franco2026conejhermitianmatricesgeometric} and \cite{lee_introduction_2012}. Some of these facts will be stated without proof as they are foundational results in differential geometry and easily verified. We will then introduce the basics of Krein spaces where our notation will largely follow \cite{franco2026conejhermitianmatricesgeometric}. More general facts about Krein spaces can be found in the foundational text \cite{Bognar1974} and \cite{Azizov1981}. It is important to note in \cite{Bognar1974} that the indefinite inner-product is introduced as an object in its own right instead of a modification of any underlying Hilbert inner-product. Hence, when verifiyng any claims by using this book, one must be extra vigilant with the changing of notation. Another important fact of which we must take heed is that in the quantum setting, the indefinite inner-product will be taken by introducing a unitary matrix in the first factor of the inner-product, which is the opposite side to that considered in some books on abstract operators. 
\item In Section \ref{Section 5}, we will consider the notion of Krein space measurements. The authors of \cite{FelipeSosa2022} allude to this notion in \cite{FelipeSosa2022} and in fact, the more general definition for $U$-spaces is made in \cite{bag2024quantumuchannelssspaces}, which naturally applies to the Krein setting. We will then introduce the definition of $J$-fidelity for Krein quantum states. We will provide the justification for this definition in a similar manner to the work done in \cite[pp.~153--154], which demonstrates the work of Fuchs and Caves in \cite{fuchs1996mathematicaltechniquesquantumcommunication} where quantum fidelity is shown to minimize the so-called Bhattacharyya coefficients of two states. We then follow this up with more observations from Chapter 3 of \cite{Watrous_2018} which follow immediately from the definition of the $J$-fidelity.
\item In Section \ref{Krein section}, we will consider the maps which behave nicely with $J$-fidelity, namely $J$-channels. Such objects have already been introduced and considered in \cite{FelipeSosa2022} and \cite{Heo2022}. We will present these from a more abstract linear algebra perspective, rather than the $C^*$-algebra presentation contained in those texts. We will show that the definitions made provide an analog of the information processing inequality. 
\item In Section \ref{Section 7}, we will consider several new definitions motivated by the understanding that a Krein space has positive and negative subspaces. There are several notions of weighted-fidelity, each with various desirable and undesirable properties. The idea of looking at such subspace motivated definitions relating to measurement has been touched upon in \cite{FelipeSosa2022} already, and here we flesh out what appropriate weighted fidelity should perhaps be.
\item In Section \ref{Section 8}, we will then repeat a lot the program of Sections \ref{Section 4}, \ref{Section 5} and \ref{Krein section} but now in the setting of Quantum $U$-channels on $S$-spaces. These are a natural generalization of $J$-channels on Krein spaces and were first considered in \cite{bag2024quantumuchannelssspaces}. We will demonstrate that $U$-fidelity makes sense with respect to the geometric notions which exist for $J$-states and we will repeat some of the program of Section \ref{Section 4} with more care, verifying the equivalent details hold in the $U$-space setting as compared with the $J$-space setting laid out in \cite{franco2026conejhermitianmatricesgeometric} and \cite{Watrous_2018} with regards to fidelity and operator means.  

\end{itemize}

\section{Notational Conventions}\label{Section 3}
We will largely keep notation as consistently as possible with \cite{Watrous_2018} since many of the results are direct analogs of theorems stated in Chapter $3$ of this book. 
In the subsequent sections, we will be considering complex finite-dimensional inner-product spaces, which we will denote with calligraphic letters $\mathcal{X},\mathcal{Y},\mathcal{Z}$ etc. The space of linear operators on such a space $\mathcal{X}$ will be denoted $L(\mathcal{X})$. The space of invertible linear transformations will be denoted $\mathrm{GL}(\mathcal{X})$ and the cones of positive definite and positive semi-definite matrices will be denoted $\mathrm{Pd}(\mathcal{X})$ and $\mathrm{Pos}(\mathcal{X})$, respectively. The set of Hermitian matrices over a space $\mathcal{X}$ will be denoted $\mathrm{Herm}(\mathcal{X})$. Quantum states will be represented by positive semidefinite operators whose traces are equal to $1$, and given a complex Euclidean space $\mathcal{X}$, we will denote this set by $D(\mathcal{X})$. 

\section{Geometric preliminaries}\label{Section 4}
\subsection{Necessary facts from the geometry of Positive Definite Matrices.}
The following recalls some useful facts from \cite{franco2026conejhermitianmatricesgeometric}.
\begin{lemma}\cite{franco2026conejhermitianmatricesgeometric}
    The cone of positive definite matrices $\mathrm{Pd}(\mathcal{X})$ forms a smooth open submanifold of $L(\mathcal{X})$. For each $P\in \mathrm{Pd}(\mathcal{X})$, the tangent space at $P$, $T_P$ can be given the inner-product 
    \begin{equation*}\langle U, V\rangle_P=\mathrm{Tr}(P^{-1}UP^{-1}V\rangle\end{equation*}
\end{lemma}
which turns $\mathrm{Pd}(\mathcal{X})$ into an open Riemannian submanifold of the space of all $n\times n$ complex matrices. 
\begin{definition}\cite{franco2026conejhermitianmatricesgeometric}
    \item The matrix exponential $\exp:L(\mathcal{X})\to GL(\mathcal{X})$ for $A\in L(\mathcal{X})$ is defined as
    \begin{equation*}A\mapsto \sum_{k=0}^\infty \frac{A^k}{k!}\end{equation*} and is a bijection to $\mathrm{Pd}(\mathcal{X})$ from $\mathrm{Herm}(\mathcal{X})$. 
    \item The inverse $\exp^{-1}:\mathrm{Pd}(\mathcal{X})\to \mathrm{Herm}(\mathcal{X})$ is the unique function \begin{equation*}\exp^{-1}=\log:\mathrm{Pd}(\mathcal{X})\to \mathrm{Herm}(\mathcal{X})\end{equation*} 
\end{definition}
\begin{lemma}\cite{franco2026conejhermitianmatricesgeometric}
    Given $P,Q\in \mathrm{Pd}(\mathcal{X})$, then unique geodesic between $P$ and $Q$ in $\mathrm{Pd}(\mathcal{X})$ is given by $\gamma:[0,1]\to \mathrm{Pd}(\mathcal{X})$,
    \begin{align*}
        \gamma(t)=\sqrt{P}\left(\exp(t\log(P^{-\frac{1}{2}}QP^{-\frac{1}{2}}))\sqrt{P}=\sqrt{P}(P^{-\frac{1}{2}}QP^{-\frac{1}{2}}\right)^t\sqrt{P}
    \end{align*}
\end{lemma}
When $t=\frac{1}{2}$ we obtain the geometric mean of the operators $P$ and $Q$. It was shown by Fuchs and Caves in \cite{fuchs1996mathematicaltechniquesquantumcommunication} that ``measuring'' in this eigenbasis achieves the minimal value of the overlap of the corresponding Bhattacharyya coefficients, so achieves the optimal fidelity between states.

\subsection{Finite-dimensional Krein spaces}
The following subsection follows from \cite{franco2026conejhermitianmatricesgeometric}
. A Krein space, in general, is an indefinite inner-product space, where vectors are allowed to have non-positive inner-products with themselves. This is achieved by an involution, normally denoted $J\in L(\mathcal{X})$. There is also the majorant topology property of Krein spaces, but this is irrelevant for us and will not come up, so we take the following to be our definition. 
\begin{definition}\cite{franco2026conejhermitianmatricesgeometric}
A Krein space $(\mathcal{X},[\cdot,\cdot]_J)$ is an inner-product space $(\mathcal{X},\langle\cdot,\cdot\rangle)$ along with a self-adjoint involution $J\in L(\mathcal{X})$ where \begin{equation*}[\cdot,\cdot]_J=\langle J \cdot,\cdot\rangle\end{equation*} and such that $[\cdot,\cdot]$ is nondegenerate. 
\end{definition}

When we say terms such as ``Krein matrix", we mean a matrix as in the above definition.

Since every space we consider is finite-dimensional, we can replace inner-product with Hilbert in the above definition.\\
In the Krein space, we define $J$-multiplication as follows; \begin{equation*}\forall A,B\in L(\mathcal{X}), A\bullet B=AJB\end{equation*} 
We list several important definitions. 
\begin{definition}\cite{franco2026conejhermitianmatricesgeometric} Let $\mathcal{X}$ be a complex Euclidean space and $J$ a Krein matrix. 
\begin{itemize}
        \item The cone of $J$-positive definite matrices, $\mathrm{Pd}_J(\mathcal{X})$ is defined as the set of $A\in L(\mathcal{X})$ such that 
        \begin{equation*}[Ax,x]_J>0\end{equation*} for all $x\in \mathcal{X}$ nonzero. 
          \item The cone of $J$-positive semi-definite matrices, $\mathrm{Pos}_J(\mathcal{X})$ is defined as the set of $A\in L(\mathcal{X})$ such that 
        \begin{equation*}[Ax,x]_J\geq 0\end{equation*} for all $x\in \mathcal{X}$. 
    \end{itemize}
\end{definition}
\begin{definition}\cite{franco2026conejhermitianmatricesgeometric}
 Let  $A\in L(\mathcal{X})$, 
  \begin{itemize}
  \item The $J$-adjoint of $A$ is the matrix \begin{equation*}A^\#=JA^\#J\end{equation*} which satisfies 
    \begin{equation*}[Ax,y]_J=[x,A^*y]_J\end{equation*} for all $x,y\in \mathcal{X}$.
    \item The $J$-inverse of $A$ (where $A$ is invertible) is defined as 
    \begin{equation*}A_J^{-1}=JA^{-1}J\end{equation*}
    \item $A$ is $J$-Hermitian when $A^\#=JA^*J=A$
    \end{itemize}
\end{definition}
\begin{remark}We notice that any $J$-positive (semi)definite operator can be made into a positive (semi)definite operator upon multiplication by $J$. This corresponds to 
``untwisting" the negative part of the cone of $J$-positive semidefinite matrices. \end{remark}

It is well known that the cone of positive definite matrices is an open sub-manifold of the space of $n\times n$ complex matrices $M_n(\mathbb{C}))$. Since $J$ is a unitary matrix, multiplication by $J$ is a diffeomorphism of $M_n(\mathbb{C})$ to itself and hence restricting to the cone of positive definite matrices is a diffeomorphism onto its image, the cone of $J$-positive definite matrices.

In \cite{franco2026conejhermitianmatricesgeometric}, analogous maps to the matrix exponential and the logarithm are defined as follows:
\begin{definition}\cite{franco2026conejhermitianmatricesgeometric}
\begin{itemize}Let $\mathcal{X}$ be a complex Euclidean space and $J$ a Krein matrix.
   \item The $J$-exponential map is defined as 
    \begin{equation*}\exp_J(X)=J\exp(JX)\end{equation*} 
    \item The $J$-logarithm is the inverse of the $J$-exponential map 
    \begin{align*}\log_J&=\exp_J^{-1}:\mathrm{Pd}_J(\mathcal{X})\to \mathrm{Herm}_J(\mathcal{X})\\
    &\log_J(A)=J\log(JA)
    \end{align*}
\end{itemize}
\end{definition}

Given a space $\mathcal{X}$, the space $\mathrm{Herm}(\mathcal{X}))$ is a finite-dimensional subspace isomorphic to $\mathbb{R}^{n^2}$ and is an embedded sub-manifold. Since multiplication by unitary matrices is a global diffeomorphism of $L(\mathcal{X})$ to itself, restricting this multiplication to $\mathrm{Herm}(\mathcal{X}))$ is a diffeomorphism onto its image, $\mathrm{Herm}_J(\mathcal{X})$ which also gains the structure of an embedded submanifold of $L(\mathcal{X})$. Hence multiplication by $J$ is a global diffeomorphism from the $\mathrm{Herm}_J(\mathcal{X})$ to $\mathrm{Herm}(\mathcal{X})$. Since the matrix exponential is a diffeomorphism from $\mathrm{Herm}_J(\mathcal{X})$ to $\mathrm{Pd}(\mathcal{X})$ and multiplication by $J$ is a diffeomorphism from $\mathrm{Pd}_J(\mathcal{X})$, we can conclude that the $J$-exponential is a global diffeomorphism from $\mathrm{Herm}_{J}(\mathcal{X})$. (For basic facts about embedded submanifolds see Chapter 5 of \cite{lee_introduction_2012}).\\
From the above we obtain the following commutative diagram.

\begin{equation*}
\begin{tikzcd}[column sep=large, row sep=large]
\mathrm{Herm}(\mathcal{X})
  \arrow[r, shift left=1.2ex, "\exp"]
&
\mathrm{Pd}(\mathcal{X})
  \arrow[l, shift left=1.2ex, "\log"]
\\
\mathrm{Herm}_{J}(\mathcal{X})
  \arrow[u, "J(\cdot)"]
  \arrow[r, shift left=1.2ex, "\exp_{J}"]
&
\mathrm{Pd}_{J}(\mathcal{X})
  \arrow[l, shift left=1.2ex, "\log_{J}"]
  \arrow[u, "J(\cdot)"]
\end{tikzcd}
\end{equation*}
 Using these maps, we can define general $J$-powers of a matrix $A\in\mathrm{Pd}_J(\mathcal{X})$ by 
\begin{equation*}X_J^t=\exp_J(t\log_J(X)=J(JX)^t\end{equation*}
In \cite{franco2026conejhermitianmatricesgeometric}, using the map $\Phi_J:\mathrm{Pd}(\mathcal{X})\to \mathrm{Pd}(\mathcal{X})$, $X\mapsto JX$, a Riemannian metric on $\mathrm{Pd}_J(\mathcal{X})$ is constructed by taking the pullback metric $\omega=\Phi^*_J(\langle\cdot,\cdot\rangle)$, which satisfies at $P\in \mathrm{Pd}_J(\mathcal{X})$, 
\begin{equation*}\omega_P(U,V)=\langle JU,JV\rangle_{JP}=\mathrm{Tr}(P^{-1}UP^{-1}V)\end{equation*}
It is then shown that given $A,B\in \mathrm{Pd}_J(\mathcal{X})$, the map $\gamma:[0,1]\to \mathrm{Pd}_J(\mathcal{X})$ given by 
\begin{equation*}\gamma(t)= A_J^\frac{1}{2}\bullet(A_J^{-\frac{1}{2}}\bullet B\bullet A_J^{-\frac{1}{2}})^t_J\bullet A_J^{\frac{1}{2}}\end{equation*} is the equation of the geodesic between $A$ and $B$ in $\mathrm{Pd}_J(\mathcal{X})$ with respect to $\omega$. When $t=\frac{1}{2}$, this is called the $J$-geometric mean of $A$ and $B$.\\
In the subsequent section, we will show that for ``$J$-quantum states", $P$ and $Q$, considering the $J$-geometric mean of $P$ and $Q^{-1}$ provides the appropriate eigenbasis for defining the fidelity between states.

\section{Fidelity between J-quantum  states}\label{Section 5}

\subsection{Krein Space States and Measurements}

\begin{definition}\cite{FelipeSosa2022}.
    Given a Krein space $(\mathcal{X}, J)$, a $J$-quantum state or Krein quantum state is an operator $P\in L(\mathcal{X})$ such that $JP$ is a quantum state.
\end{definition}
In the regular quantum theory, the notion of measurement is central to the concept of fidelity. Given a complex Euclidean space $\mathcal{X}$ and alphabet $\Sigma$, a measurement $\mu:\Sigma \to \mathrm{Pos}(\mathcal{X})$ which satisfies 
\begin{equation*}\Sigma_{a\in \Sigma} \mu(a)=\mathbf{1}_{\mathcal{X}}\end{equation*}   This definition is made with the understanding that if $P\in D(\mathcal{X})$ represents the state of a register $X$, if an element of $\Sigma$ is selected at random, the probability distribution that describes this random selection is the probability vector $p$ in $\mathbb{R}^{\Sigma}$ such that
\begin{equation*}p(a)=\langle \mu(a), P\rangle\end{equation*} \cite[p.~101]{Watrous_2018}.
Using this formalism, we make the analogous notion for $J$-quantum states.

\begin{definition}
    Given a finite dimensional Euclidean space $\mathcal{X}$, an alphabet $\Sigma$ and a fundamental symmetry $J$ on $\mathcal{X}$, a $J$-measurement is a function 
    \begin{equation*}\mu:\Sigma\to \mathrm{Pos}_J(\mathcal{X})\end{equation*}
    such that $\sum_{a\in\Sigma}\mu(a)=J$.
\end{definition}
We make this definition with the analogous measurement interpretation. If the space Krein space $\mathcal{X}$ is in the state $\rho$ when an element of $\Sigma$ is selected at random, the probability distribution that describes this selection is given by $p\in \mathcal{P}(\Sigma)$ where 
\begin{equation*}p(a)=\langle \mu(a), \rho\rangle = \langle J\mu(a),J\rho\rangle  \end{equation*}
are the entries of the probability vector.

\subsection{Definition of \texorpdfstring{$J$}{J}-fidelity}
\subsubsection{Motivation for definition}
In quantum information, fidelity is motivated by the following: given density operators $P,Q\in D(\mathcal{X})$, we aim to minimize the value of 
\begin{equation*}\sum_{a\in \Sigma} \sqrt{p(a)}\sqrt{q(a)}\end{equation*} 
over all possible measurements where $a$ belongs to the alphabet $\Sigma$, and $p,q\in \mathcal{P}(\Sigma)$, and are the vectors produced when we take a particular measurement of a space $\mathcal{X}$ which is in the state $P$ and $Q$. The values $\sqrt{p(a)}\sqrt{q(a)}$ are the Bhattacharyya coefficients \cite{Watrous_2018}. It was shown by Fuchs and Caves \cite{fuchs1996mathematicaltechniquesquantumcommunication} that the optimal measurement for achieving the fidelity is to measure in the eigenbasis of the geometric mean of the states $P$ and $Q^{-1}$; \begin{equation*}Q^{-\frac{1}{2}}\sqrt{\sqrt{Q}P\sqrt{Q}}Q^{-\frac{1}{2}}\end{equation*} from which the fidelity is defined as 
\begin{equation*}\mathrm{Tr}\left(\sqrt{\sqrt{Q}P\sqrt{Q}}\right)\end{equation*}

Hence our, definition of $J$-fidelity should be, given two $J$-states $P$ and $Q$ (operators $P,Q\in \mathrm{Pos}_J(\mathcal{X})$ such that $JP,JQ$ are density operators \cite{FelipeSosa2022}), the minimization over all possible $J$-measurements $\mu$ of 
\begin{equation*}\sum_{a\in \Sigma}\sqrt{\langle \mu(a), P\rangle}\sqrt{\langle \mu(a),Q\rangle}\end{equation*}
We will denote this value $F_J(P,Q)$.
\begin{theorem}
    The $J$-fidelity is obtained by the value 
    \begin{equation*}F_J(P,Q)=F(JP,JQ)=\mathrm{Tr}\left(\sqrt{\sqrt{{JQ}}JP\sqrt{{JQ}}}\right)\end{equation*}
\end{theorem}
\begin{remark}\label{Fidelity remark}
A measurement that achieve fidelity, in the invertible case, is found by considering the eigenbasis of the $J$-geometric mean of $P$ and the $J$ inverse of $Q$ in $\mathrm{Pd}_J(\mathcal{X})$, which is given by \begin{equation}\label{J-fidelity matrix} (Q)^{-\frac{1}{2}}_J\bullet\sqrt[J]{(\sqrt[J]{Q})\bullet P\bullet\sqrt[J]{Q}}\bullet (Q)^{-\frac{1}{2}}_J\end{equation}

\end{remark}
\begin{proof}
We follow the program contained in \cite[pp.~153--155]{Watrous_2018} and the program of proof contained in \cite{wilde2015lecture16} to verify the minimization. But, we first verify that the matrix \ref{J-fidelity matrix}  is indeed $J$-positive.\\
    Assume first of all that $P$ and $Q$ are invertible. We see that the the above matrix defined is $J$-positive semi-definite since, if we expand this expression, we obtain
    \begin{align*}
       &(Q)_J^{-\frac{1}{2}}JJ\sqrt{JJ\sqrt{JQ}JPJJ\sqrt{JQ}}J(Q)_J^{-\frac{1}{2}}\\
       =& (Q_J^{\frac{1}{2}})_J^{-1}\sqrt{\sqrt{JQ}JP \sqrt{JQ}}J(Q_J^{\frac{1}{2}})_J^{-1}\\
       =& J(J\sqrt{JQ})^{-1}J\sqrt{\sqrt{JQ}JP \sqrt{JQ}}JJ(J\sqrt{JQ})^{-1}J\\
       =& J(\sqrt{JQ})^{-1}\sqrt{\sqrt{JQ}JP\sqrt{JQ}}(\sqrt{JQ})^{-1}
    \end{align*}
    and since 
    \begin{equation*}(\sqrt{JQ})^{-1}\sqrt{\sqrt{JQ}JP\sqrt{JQ}}(\sqrt{JQ})^{-1}\end{equation*} 
    is positive semi-definite, we have that the above expression is $J$-positive semi-definite, hence multiplying by $J$ we obtain an operator $R$ which admits a traditional spectral decomposition 
    \begin{align*}
        R=\sum_{i}\lambda_ix_ix_i^*
    \end{align*}
    We then show that the optimal measurement is with respect to $x_ix_i^*$. We first note that $RJQR=JP$. This comes from the following 
    \begin{align*}
        JQ\left(\sqrt{JQ}\right)^{-1}\sqrt{\sqrt{JQ}JP\sqrt{JQ}}\left(\sqrt{JQ}\right)^{-1}= \sqrt{JQ}\sqrt{\sqrt{JQ}JP\sqrt{JQ}}\left(\sqrt{JQ}\right)^{-1}
    \end{align*}
    From this we obtain
    \begin{align*}
        RJQR&=\left(\sqrt{JQ}\right)^{-1}\sqrt{\sqrt{JQ}JP\sqrt{JQ}}\left(\sqrt{JQ}\right)^{-1}\sqrt{JQ}\sqrt{\sqrt{JQ}JP\sqrt{JQ}}\left(\sqrt{JQ}\right)^{-1}\\
        &=\left(\left(\sqrt{JQ}\right)^{-1}\sqrt{\sqrt{JQ}JP\sqrt{JQ}}\right)\left(\sqrt{\sqrt{JQ}JP\sqrt{JQ}}\left(\sqrt{JQ}\right)^{-1}\left(\sqrt{JQ}\right)^{-1}\right)\\
    \end{align*}
    Multiplying the square roots out this becomes
    \begin{align*}
        \left(\sqrt{JQ}\right)^{-1}\left(\sqrt{JQ}JP\sqrt{JQ}\right)\left(\sqrt{JQ}\right)^{-1}=JP
    \end{align*}
    
    We then have

    \begin{align*}
        &\sum_{i}\sqrt{\mathrm{Tr}((Jx_ix_i^*)^*P)\mathrm{Tr}((Jx_ix_i^*)^*Q)}\\=& \sum_{i}\sqrt{\mathrm{Tr}(x_ix_i^*JP)\mathrm{Tr}((x_ix_i^* JQ)}\\
        =& \sum_i \sqrt{(x_i^*RJQRx_i)(x_i^*JQx_i)}\\
        =& \sum_i \sqrt{(x_i^*\lambda_i JQ\lambda_ix_i)(x_i^*JQx_i)}\\
        =& \sum_i \lambda_ix_i^*JQx_i\\
        =&\mathrm{Tr}(\sum_{i}\lambda_ix_ix_i^*JQ)\\=&\mathrm{Tr}(RJQ)\\
        =&\mathrm{Tr}(\sqrt{\sqrt{JQ}JP\sqrt{JQ}})\\
        =&F(JP,JQ)
    \end{align*}
    If $\mu:\Sigma\to \mathrm{Pos}_J(\mathcal{X})$ is a $J$-measurement, then $J\mu$ is a positive semidefinite operator and hence by the regular quantum theory \cite[pp.~153--155]{Watrous_2018}, we have 
    \begin{equation*}
        F(JP,JQ)\leq \sum_{a\in \Sigma}\sqrt{\langle J\mu(a), JP\rangle}\sqrt{\langle J\mu(a),JQ\rangle}=\sum_{a\in \Sigma}\sqrt{\langle \mu(a), P\rangle}\sqrt{\langle \mu(a),Q\rangle}
    \end{equation*}
    Hence, in the case $P$ and $Q$ are invertible, we have the $J$-fidelity is defined as 
    \begin{equation*}F_J(P,Q)=\mathrm{min}_{\mu}\sum_{a}\sqrt{\langle \mu(a),P\rangle}\sqrt{\langle \mu(a), Q\rangle}=F(JP, JQ)\end{equation*}
    where this minimum is over all possible measurements.
     In the case $Q$ is not invertible, we replace $JP$ with $\Pi_{\mathrm{im}{JQ}}JP\Pi_{\mathrm{im}{JQ}}$ where we project onto the support of $JQ$. 
\end{proof}

 \subsection{Consequences of the definition of Krein Space fidelity}\label{conseqeunces of definiton}

The consequences of the above reasoning state in order to measure the ``overlap" of $J$-quantum states using a $J$-measurement, we may as well ``untwist" the cone of $J$-positive semi-definite matrices to make it the cone of positive semi-definite matrices and perform the usual optimal measurement. In this subsection, we formally verify the usual results of fidelity hold in the $J$-setting, even for mappings between Krein spaces with different involutions. The rest of the section is entirely adapted from Chapter 3 of \cite{Watrous_2018}, where a proof is not given it is a direct transplant with the appropriate $J$-multiplications and we merely include the statements for completeness.

\begin{proposition}\label{properties of J fidelity}
   We have the following properties hold:
    \begin{itemize}
    
        \item The $J$-fidelity function is continuous at $P,Q \in \mathrm{Pos}_J(\mathcal{X})$.
        \item $F_J(\lambda P,Q)=\sqrt{\lambda}F_J(P,Q)=F_J(P,\lambda Q)$ for all real $\lambda\geq 0$.
        \item $F_J(P,Q)=F(\Pi_{\mathrm{im}(JQ)}JP\Pi_{\mathrm{im}(JQ)},  JQ)=F(JP,\Pi_{\mathrm{im}(JP)}JQ\Pi_{\mathrm{im}(JP)})$.
        \item $F_J(P,Q)\geq 0$ with equality iff $PJQ=0$
        \item $F_J(P,Q)^2\leq \mathrm{Tr}(JP)\mathrm{Tr}(JQ)$ with equality iff $P$ and $Q$ are linearly dependent.
    \end{itemize}
\end{proposition}
\begin{proof}
    See \citep[p.~140]{Watrous_2018}.
\end{proof}

\begin{proposition}
    Let $\mathcal{X}$ be a complex Euclidean space and let $v\in \mathcal{X}$ be a vector and let $P\in \mathrm{Pos}_J(\mathcal{X})$. We have that 
    \begin{align*}
        F(P, Jvv^*)=\sqrt{v^*JPv}
    \end{align*}
    and in particular we have that $F(Juu^*, Jvv^*)=|\langle u,v\rangle |$
\end{proposition}

\begin{proposition}
    Let $P,Q \in\mathrm{Pos}_J(\mathcal{X})$, then 
    \begin{align*}
        F_J(P,QJPJQ)=\langle P,Q \rangle
    \end{align*}
\end{proposition}
 \begin{proof} We follow the proof of \cite[pp.~142--143]{Watrous_2018}.
     \begin{align*}
     F_J(P,QJPJQ)&=F(JP,JQJPJQ)\\
     =&\mathrm{Tr}\left(\sqrt{\sqrt{JP}JQJPJQ\sqrt{JP}}\right)\\
     =&\mathrm{Tr}\left(\sqrt{JP}JQ\sqrt{JP}\right)\\
     =&\mathrm{Tr}(JPJQ)\\
     =&\langle JP, JQ\rangle\\
     =& \langle P,Q\rangle
     \end{align*}
 \end{proof}
\begin{lemma}[Winter's gentle measurement lemma for $J$-fidelity]
Let $\mathcal{X}$ be a complex Euclidean space with choice of Krein matrix $J$ and let $\rho$ be an operator such that $J\rho\in D(\mathcal{X})$ be a $J$-quantum state and let $P\in \mathrm{Pos}_J(\mathcal{X})$ be a positive semi-definite operator satisfying $JP\leq \mathbf{1}_{\mathcal{X}}$ and $\langle P,\rho\rangle>0$. Then
\begin{equation*}F_J\left(\rho, \frac{\sqrt{JP}J\rho J\sqrt{JP}}{\langle JP, J\rho\rangle}\right)\geq \sqrt{\langle P, \rho\rangle} \end{equation*}
\end{lemma}
\begin{proof}
    We have, by following \cite[p.~143]{Watrous_2018},
    \begin{align*}
        F_J\left(\rho,\frac{J\sqrt{JP}J\rho \sqrt{JP}}{\langle JP,J\rho\rangle}\right)=\frac{1}{\sqrt{\langle JP,J\rho\rangle}}F(J\rho, \sqrt{JP}J\rho \sqrt{JP})=\frac{\langle \sqrt{JP},J\rho\rangle}{\sqrt{\langle JP, J\rho\rangle}}
    \end{align*}
    and from the assumption on $JP$, we have $\sqrt{JP}\geq JP$ and hence $\langle \sqrt{JP}, J\rho\rangle\geq \langle JP, J\rho\rangle$. Since $\langle JP,J\rho\rangle =\langle P, \rho \rangle$, the result follows.
\end{proof}

\subsection{Characterizing Krein space fidelity}
Let $\mathcal{X}$ be a complex Euclidean space with Krein operator $J$, and let $P, Q \in \mathrm{Pos}_J(\mathcal{X})$. We have
\begin{align*}
    F_J(P,Q)
    = \max \left\{ \left| \mathrm{Tr}(X) \right| : X \in L(\mathcal{X}), \;
    \begin{bmatrix}
        JP & X \\
        X^{*} & JQ
    \end{bmatrix}
     \in \mathrm{Pos}(\mathcal{X}\oplus\mathcal{X})
    \right\}.
\end{align*}

 From this we then obtain the usual semi-definite program view of fidelity (see \citep[p.~147]{Watrous_2018}) in that the $J$-fidelity $F_J(P,Q)$ gives the optimal value of the primal problem given by maximizing 
 \begin{align*}
     \frac{1}{2}\mathrm{Tr}(X)+\frac{1}{2}\mathrm{Tr}(X^*)
 \end{align*}
 subject to 
 \begin{align*}
     \begin{bmatrix}
        JP & X \\
        X^{*} & JQ
    \end{bmatrix}\geq 0
 \end{align*} for $X\in L(\mathcal{X})$.
We may also obtain the $J$-formulation of Alberti's theorem
\begin{theorem}[$J$-version of Alberti's theorem]
    For $P,Q\in \mathrm{Pos}_J(\mathcal{X})$ we have 
    \begin{align*}
        F(P,Q)^2=\inf\{\langle [P, Y]_J[Q,Y]_J^{-1}: Y\in \mathrm{Pd}_J(\mathcal{X})]
    \end{align*}
\end{theorem}
\begin{proof}
   See \citep[p.~148]{Watrous_2018}.
\end{proof}
\subsubsection{Brief aside: The trace and partial trace}
From the standard quantum theory, we know that the trace and partial trace are invaluable tools in the consideration of states. Motivated by this, we consider $P\in \mathrm{Pos}_J(\mathcal{X})$, then there exists a vector $u\in \mathcal{X}\otimes \mathcal{Y}$ such that \begin{equation*}\mathrm{Tr}_{\mathcal{Y}}((J\otimes I)(J\otimes I)uu^*))=JP\end{equation*} iff $\dim(\mathcal{Y})\geq \mathrm{rank}(P)$. In general, given a Krein matrix $J$, we will call the function $\mathrm{Tr}_{\mathcal{Y}}((J\cdot)$ the $J$-partial trace and similarly, we consider the $J$ trace to be $\mathrm{Tr}(J\cdot)$ (remembering that $J\otimes I$ is still a Krein matrix) and we will denote the $J$-partial trace $\mathrm{Tr}^J_{\mathcal{Y}}$ and the $J$-trace $\mathrm{Tr}^J$. Using this language, we obtain an alternative definition of the $J$-fidelity.
\begin{lemma}
    The $J$-fidelity of two $J$-positive semidefinite operators $P,Q$ is given by 
    \begin{align*}
        F_J(P,Q)=\mathrm{Tr}^J\sqrt[J]{(\sqrt[J]{Q})^{\#}\bullet(\sqrt[J]{P})^{\#}\bullet\sqrt[J]{P}\bullet\sqrt[J]{Q}})
    \end{align*}
\end{lemma}
\begin{proof}
    \begin{align*}(\sqrt[J]{Q})^{\#}\bullet(\sqrt[J]{P})^{\#}\bullet\sqrt[J]{P}\bullet\sqrt[J]{Q}   &=J\sqrt{JQ}JJJJ\sqrt{JP}JJJJ\sqrt{JP}JJ\sqrt{JQ}\\
    &=J\sqrt{JQ}JP\sqrt{JQ}
    \end{align*}
    We then have that taking the $J$-square root gives 
    \begin{equation*}\sqrt[J]{J\sqrt{JQ}JP\sqrt{JQ}}=J\sqrt{\sqrt{JQ}JP\sqrt{JQ}}\end{equation*}
     and hence taking the $J$-trace gives us 
    \begin{align*}
        \mathrm{Tr}\left(JJ\sqrt{\sqrt{JP}JQ\sqrt{JQ}}\right)=\mathrm{Tr}\left(\sqrt{\sqrt{JP}JQ\sqrt{JQ}}\right)=F_J(P,Q)
    \end{align*}
\end{proof}

\begin{remark}
    This was actually the first definition of $J$-fidelity considered by the author. The reason for this was the trace definition of fidelity \[F(P,Q)=\text{Tr}(\sqrt{\sqrt{Q}P\sqrt{Q}})\]
\end{remark}
\subsection{J-analog of Uhlmann's theorem}
\begin{lemma}
    Let $A,B\in L(\mathcal{Y},\mathcal{X})$ be operators for complex Euclidean spaces $\mathcal{X}$ and $\mathcal{Y}$ and endow $\mathcal{X}$ with a Krein matrix $J$. We have 
    \begin{align*}
        F_J(JAA^*,JBB^*)=\|A^*B\|_1
    \end{align*}
\end{lemma}
\begin{proof}
   See \cite[p.~151]{Watrous_2018}.
\end{proof}
\begin{theorem}
    Let $\mathcal{X}$ and $\mathcal{Y}$ be complex Euclidean spaces and let $J_1$ be a Krein matrix for $\mathcal{X}$ and $J_2$ a Krein matrix for $\mathcal{Y}$ and let $P,Q\in \mathrm{Pos}_{J_1}(\mathcal{X})$ be $J_1$-positive semi-definite operators having rank at most $\mathrm{dim}(\mathcal{Y})$ and let $u\in \mathcal{X}\otimes {Y}$ satisfy $\mathrm{Tr}^{J_1\otimes J_2}_{\mathcal{Y}}((J_1\otimes J_2)uu^*)=J_1P$. It holds that 
    \begin{align*}
        F_J(P,Q)=\max\{|\langle u, v\rangle|:v\in \mathcal{X}\otimes \mathcal{Y}, \mathrm{Tr}_{\mathcal{Y}}^{J_1\otimes J_2}((J_1\otimes J_2)vv^*) \}
    \end{align*}
\end{theorem}
\begin{proof}
We follow the proof of \citep[p.~151--152]{Watrous_2018}.
    Let $A\in L(\mathcal{Y},\mathcal{X})$ be the operator such that $u=\mathrm{vec}(A)$ is the vectorization of $A$ (see \citep[p.~23]{Watrous_2018} and let $w\in \mathcal{X}\otimes \mathcal{Y}$ be the vector satisfying $J_1Q=\mathrm{Tr}_{\mathcal{Y}}^{J_1\otimes J_2}((J_1\otimes J_2)ww^*)$ and let $B\in L(\mathcal{Y},\mathcal{X})$ be the operator for which $w=\mathrm{vec}(B)$. We then have 
    \begin{align*}
        &\mathrm{max}\{|\langle u,v\rangle|: v\in \mathcal{X}\otimes \mathcal{Y}, \mathrm{Tr}_{\mathcal{Y}}^{J_1\otimes J_2}((J_1\otimes J_2)vv^*)=J_1Q\}\\
        =& \max\{|\langle u, (\mathbf{1}_{\mathcal{X}}\otimes U)w\rangle| :U\in U(\mathcal{Y})\}\\
        =&\max{|\langle A,BU^T\rangle|:U\in U(\mathcal{Y})}\}\\
        =& \max\{|\langle \overline{U}, A^*B\rangle|:U\in U(\mathcal{Y})\}\\
        =&\|A^*B\|_1=F_{J_1}(J_1AA^*,J_1BB^*)=F_{J_1}(P,Q)
    \end{align*}
\end{proof}
\begin{corollary}
Let $u, v \in \mathcal{X} \otimes \mathcal{Y}$ be vectors for complex Euclidean spaces $\mathcal{X}$ and $\mathcal{Y}$ equipped with Krein matrices $J_1$ and $J_2$, so that $\mathcal{X} \otimes \mathcal{Y}$ is equipped with Krein matrix $J_1 \otimes J_2$. We have
\begin{align*}
F_{J_1}\Big(
\mathrm{Tr}_{\mathcal{Y}}^{J_1 \otimes J_2}\big((J_1 \otimes J_2)(uu^*)\big),
\mathrm{Tr}_{\mathcal{Y}}^{J_1 \otimes J_2}\big((J_1 \otimes J_2)(vv^*)\big)
\Big)=\left\|\mathrm{Tr}_{\mathcal{X}}^{J_1 \otimes J_2}\big((J_1 \otimes J_2)(v u^*)\big)\right\|_1
\end{align*}
\end{corollary}

\begin{proof}
Let $A,B \in L(\mathcal{Y},\mathcal{X})$ be such that $u=\mathrm{vec}(A)$ and $v=\mathrm{vec}(B)$.

\begin{align*}
&F_{J_1}\Big(
\mathrm{Tr}_{\mathcal{Y}}^{J_1 \otimes J_2}\big((J_1 \otimes J_2)(uu^*)\big),
\mathrm{Tr}_{\mathcal{Y}}^{J_1 \otimes J_2}\big((J_1 \otimes J_2)(vv^*)\big)
\Big) \\
&= F_{J_1}(J_1 A A^*,\, J_1 B B^*) \\
&= \|A^* B\|_1 \\
&= \|(A^* B)^T\|_1 \\
&= \left\|\mathrm{Tr}_{\mathcal{X}}^{J_1 \otimes J_2}\big((J_1 \otimes J_2)(v u^*)\big)\right\|_1.
\end{align*}
\end{proof}
\section{Morphisms of Krein spaces}\label{Krein section}
We wish to define the $J$-preserving maps of our  spaces in an analogous way to the usual case. 
The first to consider is the notion of $J$-positive maps. To consider in full generality, we start in the setting that $\mathcal{X}$ and $\mathcal{Y}$ are complex Euclidean spaces each equipped with a $J$-matrix $J_1$ and $J_2$ respectively. The following subsection follows from \cite{FelipeSosa2022} and we mainly reframe terminology and results to match up to the notation of this paper and \cite{Watrous_2018}.
\begin{definition}
    A map $\Phi\in T(\mathcal{X},\mathcal{Y})$ is $(J_1,J_2)$-positive if $A\in \mathrm{Pos}_{J_1}(\mathcal{X})$ implies $\Phi(A)\in \mathrm{Pos}_{J_2}(\mathcal{Y})$. 
\end{definition}
The above definition is made in \cite{FelipeSosa2022}.
We have that if $A\in \mathrm{Pos}_{J_1}(\mathcal{X})$ then we can write $A=JB$ for $B\in \mathrm{Pos}(\mathcal{X})$ and similarly if $\Phi(A)\in \mathrm{Pos}_{J_2}(\mathcal{Y})$ then $J_2\Phi(A)\in\mathrm{Pos}(\mathcal{Y})$. From this we obtain
\begin{lemma}
    $\Phi:L(\mathcal{X})\to L(\mathcal{Y})$ is $(J_1,J_2)$-positive if and only if $J_2\Phi(J_1\cdot):L(\mathcal{X})\to L(\mathcal{Y})$ is positive.
\end{lemma}
We have the following diagram 
\begin{equation*}
\begin{tikzcd}
\mathrm{Pos}_{J_1}(\mathcal{X})\arrow[r, "\Phi"] \arrow[d, <->, "J_1(\cdot)"'] & \mathrm{Pos}_{J_2}(\mathcal{Y}) \arrow[d, <->, "J_2(\cdot)"] \\
\mathrm{Pos}(\mathcal{X}) \arrow[r, "\Psi"'] & \mathrm{Pos}(\mathcal{Y})
\end{tikzcd}
\end{equation*}
where $\Psi=J_2\Phi(J_1\cdot)$ and from this we see the action of $\Phi$ can be recovered from $\Psi$ by taking $J_2\Psi(J_1\cdot)$
Following this, we are now motivated to make the definition of completely $J$-positive maps in the most general setting.
\begin{definition}[Completely $J$-positive maps]
Given two Krein spaces $(\mathcal{X},J_1), (\mathcal{Y}, J_2)$, we say that \begin{equation*}\Phi:L(\mathcal{X})\to L(\mathcal{Y})\end{equation*} is completely $(J_1,J_2)$-positive if for any $J$ space $(\mathcal{Z}, J_3)$, we have that 
\begin{align*}
    \Phi\otimes \mathbf{1}_{L(\mathcal{Z})}:L(\mathcal{X})\otimes L(\mathcal{Z}) \to L(\mathcal{Y})\otimes L(\mathcal{Z})
\end{align*} is $(J_1\otimes J_3, J_2\otimes J_3)$ positive.
\end{definition}
We have that the above definition is very general, and as we will see below it is a little unnecessary. 
The above definition is made in \cite{FelipeSosa2022} following the definition of completely $J$-positive maps in terms of Kraus operators. It is well known that a map $\Phi$ is completely positive if and only if it admits a Kraus decomposition of the form
\begin{equation*}
    \Phi(X)=\sum_{a\in \Sigma}A_a^*XA_a
\end{equation*}
for $A_a\in L(\mathcal{X})$. Hence, as observed in \cite{FelipeSosa2022}, we have $\Phi$ is completely $J$-positive if and only if this can be written has 
\begin{equation*}
    \Phi(X)=\sum_{a\in \Sigma}A_a^\#XA_a
\end{equation*}
We have that our definition coincides with the definition as made in \cite{FelipeSosa2022}, which is equivalent to Lemma \ref{lemma 6.4} below. 
\begin{lemma}\label{lemma 6.4}
    The above definition can be characterized as \begin{equation*}\Phi:L(\mathcal{X})\to L(\mathcal{Y})\end{equation*} is completely $(J_1,J_2)$-positive if and only if 
    \begin{align*}
        J_2\Phi(J_1\cdot)\otimes \mathbf{1}_{L(\mathcal{Z})} 
    \end{align*}
    is positive for any $J$-space $(\mathcal{Z},J_3)$
\end{lemma}
\begin{proof}
     If $\Phi$ is completely $(J_1,J_2)$ positive, we then have that $(J_2\otimes J_3)(\Phi\otimes \mathbf{1}_{\mathcal{Z}})((J_1\otimes J_2)(\cdot))$ is positive. Expanding out we obtain 
    \begin{align*}
        (J_2\otimes J_3)(\Phi\otimes \mathbf{1}_{\mathcal{Z}})((J_1\otimes J_3)\cdot)=J_2\Phi(J_1\cdot)\otimes J_3\mathbf{1}_{L(\mathcal{Z})}J_3=J_2\Phi(J_1\cdot)\otimes \mathbf{1}_{L(\mathbf{Z})}
    \end{align*}
    which implies $J_2\Phi(J_1\cdot)$ is a completely positive map. Applying the equalities backwards, we obtain the desired if and only if. 
\end{proof}
We obtain that if $\Phi:L(\mathcal{X})\to L(\mathcal{Y})$ is completely $(J_1,J_2)$ positive then the following diagram commutes for every complex Euclidean $J$-space $(\mathcal{Z},J_3)$
\begin{equation*}
\begin{tikzcd}[column sep=2cm]
\mathrm{Pos}_{J_1\otimes J_3}(\mathcal{X}\otimes \mathcal{Z})\arrow[r, "\Phi\otimes \mathbf{1}_{L(\mathcal{Z})}"] \arrow[d, <->, "J_1\otimes J_3(\cdot)"'] & \mathrm{Pos}_{J_2\otimes J_3}(\mathcal{Y}\otimes \mathcal{Z}) \arrow[d, <->, "J_2\otimes J_3(\cdot)"] \\
\mathrm{Pos}(\mathcal{X}\otimes \mathcal{Z}) \arrow[r, "J_2\Psi(J_1\cdot)\otimes \mathbf{1}_{L(\mathcal{Z})}"'] & \mathrm{Pos}(\mathcal{Y}\otimes \mathcal{Z})
\end{tikzcd}
\end{equation*}
\subsection{Information processing in the Krein setting}
In the standard quantum theory, we have that applying a channel to two states increases their fidelity, a nice mathematical explanation of how noisy channels can reduce the ability to discriminate states. We have that the notion of a $J$-channel has this same effect on the $J$-fidelity. 
\begin{theorem}
    Let $P,Q\in \mathrm{Pos}_{J_1}(\mathcal{X})$ and let $\Phi:L(\mathcal{X})\to L(\mathcal{Y})$ be a $(J_1,J_2)$-channel. Then
    $F_{J_1}(P,Q)\leq F_{J_2}(\Phi(P),\Phi(Q))$
\end{theorem}
\begin{proof}
Again, we simply adapt the proof \citep[p.~156]{Watrous_2018}.
    We may choose $X\in L(\mathcal{X})$ such that 
    \begin{align*}
        \begin{bmatrix}
J_1P & X \\
X^* & J_1Q
\end{bmatrix}
    \end{align*}
    is positive semi-definite
    and such that $|\mathrm{Tr}(X)|=F_{J_1}(P,Q)$. We have that $\Phi$ is $(J_1,J_2)$-positive, hence we have 
\begin{equation*}
\begin{bmatrix}
J_2\Phi(P) & J_2\Phi(J_1X) \\
J_2\Phi(J_1X^*) & J_2\Phi(Q)
\end{bmatrix}=\begin{bmatrix}
J_2\Phi(P) & J_2\Phi(J_1X) \\
J_2\Phi(J_1X)^* & J_2\Phi(Q)
\end{bmatrix}
\end{equation*}
is positive semi-definite, where we have used that $\Phi$ being a $(J_1,J_2)$ channel is equivalent to $J_2\Phi(J_1\cdot)$ being a channel, allowing us to move the adjoint to the outside. Since $J_2\Phi(J_1\cdot)$ is trace preserving. we have that 
\begin{align*}
    F_{J_1}(\Phi(P),\Phi(Q))\geq |\mathrm{Tr}(J_2\Phi(J_1X))|=|\mathrm{Tr}(X)|=F_{J_1}(P,Q)
\end{align*}
\end{proof}
\section{Weighted Krein space fidelity}\label{Section 7}
The definition of $J$-fidelity is in some ways unattractive as it completely ignores the decomposition of the space into positive and negative parts. There are several possible notions of ``weighted" $J$-fidelity. One candidate is to take the regular trace instead of the $J$-trace in the definition of fidelity, 
\begin{equation*}F_J^W(P,Q)=\mathrm{Tr}\left(J\sqrt{\sqrt{JQ}JP\sqrt{JQ}}\right)\end{equation*}
This definition does have some benefits, especially when $JP$ and $JQ$ can be expressed in block diagonal form with respect to the decomposition of the space into positive and negative parts. Furthermore, given that $P$ and $Q$ are $J$-states, we have 
\begin{equation*}-1\leq F_J^W(P,Q)\leq 1\end{equation*} since 
\begin{align*}
    F_J^W(P,Q)&=\mathrm{Tr}\left(\Pi_1\sqrt{\sqrt{JQ}JP\sqrt{JQ}}-\Pi_2\sqrt{\sqrt{JQ}JP\sqrt{JQ}}\right)\\
               &=\mathrm{Tr}\left(\Pi_1\sqrt{\sqrt{JQ}JP\sqrt{JQ}}\right)-\mathrm{Tr}\left(\Pi_2\sqrt{\sqrt{JQ}JP\sqrt{JQ}}\right)\\
               &=\mathrm{Tr}\left((\Pi_1^2\sqrt{\sqrt{JQ}JP\sqrt{JQ}}\right)-\mathrm{Tr}\left(\Pi_2^2\sqrt{\sqrt{JQ}JP\sqrt{JQ}}\right)\\
               &=\mathrm{Tr}\left(\Pi_1\sqrt{\sqrt{JQ}JP\sqrt{JQ}}\Pi_1\right)-\mathrm{Tr}\left(\Pi_2\sqrt{\sqrt{JQ}JP\sqrt{JQ}}\Pi_2\right)
\end{align*}
where $\Pi_1$ and $\Pi_2$ are the associated projections onto the positive and negative parts of the space.
Since the terms in this subtraction are the traces of projected positive semi-definite operators, whose trace is bounded above by $1$, we then obtain the desired bound. If $|F^W_J(P,Q)|=1$, then we have 
\begin{equation*}F^W_J(P,Q)=F_J(P,Q)\end{equation*} which then implies $JP=JQ$, and hence $P=Q$. We also then must have that $JP$ is invariant under the projections $\Pi_1$ and $\Pi_2$, which implies that $JP$ is a singular positive semi-definite operator which is block diagonal with $0$s in the top left or bottom right block. Since a positive semi-definite matrix with a zero block in the top left or bottom right block also must have zero blocks on the off-diagonal blocks (see \citep[pp.~144--146]{Watrous_2018}), we know then that $JP$ must be of the form \begin{equation*}JP=\begin{pmatrix}
    A & 0\\
    0 & 0
\end{pmatrix}\end{equation*} or \begin{equation*}JP=\begin{pmatrix}
    0&0\\
    0&-A
\end{pmatrix}\end{equation*}
If the upper bound of $1$ is achieved 
\begin{align*}
1=\left|\mathrm{Tr}\left(J\sqrt{\sqrt{JQ}JP\sqrt{JQ}}\right)\right|&\leq \|J\|_\infty \left\|\sqrt{\sqrt{JQ}JP\sqrt{JQ}}\right\|_1\\&=\mathrm{Tr}\left(\sqrt{JPJQJPJQ}\right)\\&=\mathrm{Tr}(JPJQ)\\&\leq\|JP\|_1\|JQ\|_\infty\\
&\leq 1
\end{align*}

which means that these inequalities are really equalities. This leads us to the following.
\begin{theorem}\label{pure state theorem}
    $J$-quantum states $P$ and $Q$ satisfy $|F_J^W(P,Q)|=1$ if and only if $P=Q$ and $P$ is invariant under projection by onto the first or second factor in the decomposition of the space. 
\end{theorem}
\begin{proof}
The only if direction follows from the fact that if \begin{equation*}P=Q= 
\begin{pmatrix}
A & 0 \\
0 & 0
\end{pmatrix}
\end{equation*} then $F_J^W(P,Q)=\mathrm{Tr}(A)$ with an analogous result holding for \begin{equation*}P=Q=\begin{pmatrix}
0 & 0 \\
0 & A
\end{pmatrix}\end{equation*} The previous exposition shows that if $F_J^W(P,Q)=1$, we have that $P$ must be of the desired form.
\end{proof}
\begin{remark}
It was observed in conversation with Doctor Javier Alejandro Chavez-Dominguez that since the above implies that $\|P\|_{\infty}=1$, Theorem \ref{pure state theorem} also implies that $P$ must be a pure state, which seems rather remarkable as this is a very strong conclusion from a seemingly weaker assumption.
\end{remark}

So we see this definition has some nice properties, but for non block diagonal operators, the $J$ term inside the trace can behave wildly. This definition is also not commutative, which can be fixed by taking the symmetrization as the definition instead:
\begin{equation*}F_J^{SW}(P,Q)=\frac{1}{2}\mathrm{Tr}(J\sqrt{\sqrt{JQ}JP\sqrt{JQ}})+\frac{1}{2}\mathrm{Tr}((J\sqrt{\sqrt{JP}JQ\sqrt{JP}})\end{equation*}
This can then be further simplified using the polar decompositions of $\sqrt{JQ}\sqrt{JP}$ and $\sqrt{JP}\sqrt{JQ}$ which provide unitary $U,V\in U(\mathcal{X})$ such that 
\begin{equation*}F_J^{SW}=\frac{1}{2}\langle UJ,\sqrt{JP}\sqrt{JQ}\rangle+\frac{1}{2}\langle VJ,\sqrt{JQ}\sqrt{JP}\rangle\end{equation*}

In keeping with the fact that any Hilbert space can be embedded into a larger Krein space by just extending every vector by zeros with respect to some basis and declaring the additional dimensions to be the  ``negative part" \cite{Bognar1974}, we can view any measurement as being really some measurement taken on the positive part of some larger Krein space. If we decompose a Krein space $\mathcal{X}$ as 
\begin{equation*}\mathcal{X}=\mathcal{X}_1\oplus\mathcal{X}_2\end{equation*} where $\mathcal{X}_1$ and $\mathcal{X}_2$ are the respective positive and negative parts, then for any $J$-quantum state $X\in \mathrm{Pos}_{J}(\mathcal{X})$, we can take quantum measurements of 
\begin{equation*}\Pi_1JX\Pi_1\in L(\mathcal{X}_1),\,\,\,\Pi_2JX\Pi_2\in L(\mathcal{X}_2)\end{equation*}
respectively where $\Pi_1$ and $\Pi_2$ are the orthogonal projections onto $\mathcal{X}_1$ and $\mathcal{X}_2$ respectively. For $X,Y\in L(\mathcal{X})$, we can consider the fidelity of their projections and consider the difference between these two values to provide a notion of weighted fidelity, we will call this measurement weighted fidelity and define it as 
\begin{equation*}F_J^M(P,Q)=\mathrm{Tr}\left(\sqrt{\sqrt{\Pi_1JQ\Pi_1}\Pi_1JP\Pi_1\sqrt{\Pi_1JQ\Pi_1}}\right)-\mathrm{Tr}\left(\sqrt{\sqrt{\Pi_2JQ\Pi_2}\Pi_2JP\Pi_2\sqrt{\Pi_2JQ\Pi_2}}\right)\end{equation*}
This definition is commutative, satisfies the same bounds, and is easier to handle in general because the decomposition of the space is taken into account inside the square root. However, we lose the off diagonal blocks of $P$ and $Q$ in the $J$ basis. We have that this definition behaves nicer with already known properties of $J$-fidelity.
\subsection{Elementary properties of weighted measurement fidelity}
We collect a few results by applying regular fidelity properties contained \cite{Watrous_2018} to this notion of weighted measurement fidelity.
\begin{lemma}
    Given $P,Q\in\mathrm{{Pos}}_{J}(\mathcal{X})$, we have the following holds
    \begin{itemize}
        \item The measurement weighted fidelity is continuous at $(P,Q)$.
        \item $F_J^M(\lambda P,Q)=F_J^M(P,\lambda Q)=\sqrt{\lambda}F_J^M(P,Q)$ for all real $\lambda\geq 0$.
        \item $F_J^M(P,Q)=\|\sqrt{\Pi_1JP\Pi_1}\sqrt{\Pi_1JQ\Pi_1}\|_1-\|\sqrt{\Pi_2JP\Pi_2}\sqrt{\Pi_2JQ\Pi_2}\|_1$
        \item $-1\leq F_J^M(P,Q)\leq 1$.
    \end{itemize}

\end{lemma}
\begin{proposition}
    Let $\mathcal{X}$ be a complex Euclidean space with Krein matrix $J$ and let $P,Q\in \mathrm{Pos}_J(\mathcal{X})$ and let $\mathcal{X}=\mathcal{X}_1\oplus\mathcal{X}_2$ be the decomposition of the space into positive and negative parts. It holds that 
    \begin{align*}F_J^M(P,Q)=&\max\{|\mathrm{Tr}(X)|:X\in L(\mathcal{X}_1), \begin{pmatrix}
        \Pi_1JP\Pi_1& X\\
        X^*& \Pi_1JQ\Pi_1
    \end{pmatrix}\in\mathrm{Pos}(\mathcal{X}_1\oplus\mathcal{X}_1)\}\\-&\max\{|\mathrm{Tr}(Y)|:Y\in L(\mathcal{X}_2), \begin{pmatrix}
        \Pi_2JP\Pi_2 & Y\\
        Y^* & \Pi_2JQ\Pi_2
    \end{pmatrix} \in \mathrm{Pos}(\mathcal{X}_2\oplus\mathcal{X}_2)\}
    \end{align*}
\end{proposition}
\begin{proof}\label{weighted measurement fidelity}
    We have by \citep[p.~144]{Watrous_2018} that 
    \begin{equation*}F(\Pi_1JP\Pi_1,\Pi_1JQ\Pi_1)=\max\left\{|\mathrm{Tr}(X)|:X\in L(\mathcal{X}_1), \begin{pmatrix}
        \Pi_1JP\Pi& X\\
        X^*& \Pi_1JQ\Pi_1
    \end{pmatrix}\in\mathrm{Pos}(\mathcal{X}_1\oplus\mathcal{X}_1)\right\}\end{equation*}
    and 
    \begin{equation*}F(\Pi_2JP\Pi_2,\Pi_2JP\Pi_2)=\max\left\{|\mathrm{Tr}(Y)|:Y\in L(\mathcal{X}_2), \begin{pmatrix}
        \Pi_2JP\Pi_2 & Y\\
        Y^* & \Pi_2JQ\Pi_2
    \end{pmatrix} \in \mathrm{Pos}(\mathcal{X}_2\oplus\mathcal{X}_2)\right\}\end{equation*}
    and subtracting these values gives us the weighted measurement fidelity.
\end{proof}
\begin{theorem}[Uhlmann's Theorem for weighted measurement fidelity]
Let $\mathcal{X}$ be a complex Euclidean space with Krein matrix $J$. Let $P,Q\in \mathrm{Pos}_J(\mathcal{X})$. Let $\mathcal{X}=\mathcal{X}_1\oplus\mathcal{X}_2$ be the fundamental decomposition of $\mathcal{X}$ with respect to $J$. Let $\mathcal{Y}$ be a complex Euclidean space having $\dim(\mathcal{Y})\geq \max\{\rank(P),\rank(Q)\}$. Then there exists an orthogonal decomposition of $\mathcal{Y}=\mathcal{Y}_1\oplus\mathcal{Y}_2$ and $u_1\in \mathcal{X}_1\otimes\mathcal{Y}_1$ and $u_2\in \mathcal{X}_2\otimes \mathcal{Y}_2$ such that $\mathrm{Tr}_{\mathcal{Y}_1}(u_1u_1^*)=\Pi_1JP\Pi_1$ and $\mathrm{Tr}_{\mathcal{Y}_2}(u_2u_2^*)=\Pi_2JP\Pi_2$ such that 
\begin{align*}
    F_J^M(P,Q)=&\max\{|\langle u_1,v_1\rangle|: v_1\in \mathcal{X}_1\otimes \mathcal{Y}_1, \mathrm{Tr}_{\mathcal{Y}_1}(v_1v_1^*)=\Pi_1JQ\Pi_1\}\\-&\max\{|\langle u_2,v_2\rangle|: v_2\in \mathcal{X}_2\otimes \mathcal{Y}_2,\mathrm{Tr}_{\mathcal{Y}_1}(v_2v_2^*)=\Pi_2JQ\Pi_2\}
\end{align*}
     
\end{theorem}
\begin{proof}
    The proof is just carefully splitting up the proof of Uhlmann's theorem in \citep[p.~151--152]{Watrous_2018}.
\end{proof}

\section{Quantum \texorpdfstring{$U$}{U}-channels on \texorpdfstring{$S$}{S}-spaces:}\label{Section 8}
The notion of Krein quantum states has recently been extended in \cite{bag2024quantumuchannelssspaces} to complex Euclidean spaces, where instead of $J$, we use a general unitary to induce an indefinite inner-product on the underlying Hilbert space. Such a space is called an $S$-space. We show that we the expected notion of fidelity from the Krein case extends naturally from to the case of the $S$-space. In this section all underlying spaces are complex Euclidean.
\subsection{Some basic definitions and properties of S-spaces:}
\begin{definition}
    Let $\mathcal{X}$ be a Hilbert space and let $U\in U(\mathcal{X})$. The space $\mathcal{X}$ endowed with the sesquilinear form
    \begin{align*}
        [x,y]_U=\langle U^*x ,y\rangle\,\,\,\forall x,y\in \mathcal{X} 
    \end{align*}
    is called an $S$-space\cite{bag2024quantumuchannelssspaces}.
\end{definition}
Analogously to the Krein case, there is a natural notion of the adjoint map
\begin{definition}[$U$-adjoint.]\cite{bag2024quantumuchannelssspaces}
    For every $V\in L(\mathcal{X})$, there exists an operator $V^\#=U^*V^*U$ such that 
    \begin{align*}
        [Vx,y]_U=[x,V^\#y]_U
    \end{align*}
\end{definition}
Indeed, we have 
\begin{equation*}[x,V^\#]_U=\langle U^*x,U^*V^*Uy\rangle=\langle x,V^*Uy\rangle=\langle U^*Vx,y\rangle=[Vx,y]_U \end{equation*}

\begin{definition}[Cone of $U$-positive semi-definite matrices.]
A matrix $P \in L(\mathcal{X})$ is $U$ positive if $U^*P\in \mathrm{Pos}(\mathcal{X})$\cite{bag2024quantumuchannelssspaces}. We let $\mathrm{Pos}_U(\mathcal{X})$ denote the cone of $U$-positive semi-definite matrices.

\end{definition}
Using the sesquilinear form given by $U$, we can endow a different geometry on $\mathrm{Pos}_U(\mathcal{X})$ by following a similar program to \cite{franco2026conejhermitianmatricesgeometric}, in which multiplication is given by \begin{equation*}A\bullet B=AU^*B\end{equation*} With this definition of multiplication, the matrix $U$ acts as the identity since 
\begin{equation*}U\bullet A=UU^*A=AU^*U=A\end{equation*}
Using this definition of multiplication, we may also define several other important notions 
\begin{definition}[$U$-inverse]
A matrix $A\in L(\mathcal{X})$ is invertible if and only if there exists a $B\in L(\mathcal{X})$ such that $A\bullet B=B\bullet A=U$. This matrix $B$ must be given by $B=UA^{-1}U$
\end{definition}
\begin{definition}[$U$-square root]
Given $A\in \mathrm{Pos}_U(\mathcal{X})$, we define the $U$-square root to be 
\begin{equation*}A_U^{\frac{1}{2}}=U\sqrt{U^*A}\end{equation*}
\end{definition}
We see that the above definition then satisfies 
\begin{equation*}A_U^{\frac{1}{2}}\bullet A_U^{\frac{1}{2}}=U\sqrt{U^*A}U^*U\sqrt{U^*A}=U(U^*A)=A\end{equation*}
\begin{proposition}
    For all invertible $A\in L(\mathcal{X})$, we have $APA^\#\in \mathrm{Pos}_U(\mathcal{X})$. In other words, the group of $|\Sigma|\times |\Sigma|$ invertible matrices in $L(\mathcal{X})$ acts as a transitive group action by $U$-conjugation on $\mathrm{Pos}_{U}(\mathcal{X})$. 
\end{proposition}

\begin{definition}
    We define the $U$-exponential as the map 
    $\exp_U(X)=U\exp(U^*X)$ which restricts to a bijection 
    \begin{equation*}\exp_U(X):\mathrm{Herm}_U(\mathcal{X})\to \mathrm{Pd}_U(\mathcal{X})\end{equation*}
\end{definition}
The definition of $\mathrm{Herm}_U(\mathcal{X})$ is obviously the $X\in L(\mathcal{X})$ invariant under $U$-conjugation. From the above we can also define $\log_U(X)$ to be the inverse of $\exp_U$ when we restrict $\exp_U$ to $\mathrm{Herm}_U(\mathcal{X})$. We can also consider the quantity \begin{equation*}\exp_U(t\log_U(X))=X_U^t\end{equation*}, as in the $J$ case, for $X\in \mathrm{Pd}_U(\mathcal{X})$
\subsection{Geodesics in \texorpdfstring{$\mathrm{Pd}_U(\mathcal{X})$}{PdU(X)}.}
In this subsection, we repeat some of the program carried out in \cite{franco2026conejhermitianmatricesgeometric}, but now in the $S$-space setting.

Let $\Phi_U:\mathrm{Pd}_U(\mathcal{X})\to \mathrm{Pd}(\mathcal{X})$ be the map $\Phi_U(X)=U^*X$ and define $\omega^U=\Phi^*_U(\langle\cdot,\cdot\rangle)$. This then satisfies for all $P\in \mathrm{Pd}_U(\mathcal{X})$
\begin{equation*}\omega^U_P(R,S)=\Phi^*(\langle R,S\rangle)=\langle U^*R,U^*S\rangle_{U^*P} \end{equation*} for $R,S \in \mathrm{Herm}_{U}(\mathcal{X})$. Indeed, we obtain
\begin{align*}
    \langle U^*R,U^*S\rangle_{U^*P}&=\mathrm{Tr}((U^*P)^{-1}U^*R(U^*P)^{-1}U^*S)\\
    &=\mathrm{Tr}(P^{-1}UU^*RP^{-1}UU^*S)\\
    &=\mathrm{Tr}(P^{-1}RP^{-1}S)
\end{align*}

\begin{theorem}
    Let $A,B\in \mathrm{Pd}_U(\mathcal{X})$. The map $\gamma:[0,1]\to \mathrm{Pd}_U(\mathcal{X})$ given by 
    \begin{equation*}\gamma(t)=A_U^{\frac{1}{2}}\bullet ((A_U^{-\frac{1}{2}})\bullet B\bullet (A_U^{-\frac{1}{2}}))_U^t\bullet A_U^{\frac{1}{2}}\end{equation*} is the equation of the geodesic between $A$ and $B$ in $\mathrm{Pd}_U(\mathcal{X})$ with respect to $\omega^U$. 
\end{theorem}
\begin{proof}
    Let $\delta$ be the equation of the geodesic between $U^*A$ and $U^*B$ in $\mathrm{Pd}_{U}(\mathcal{X})$. This is the curve 
    \begin{equation*}\delta(t)=\sqrt{U^*A}((\sqrt{U^*A})^{-1}U^*B(\sqrt{U^*A})^{-1})^t\sqrt{U^*A}\end{equation*}
    Since multiplication by $U$ on $\mathrm{Pd}_{U}(\mathcal{X})$ is an isometry of Riemannian manifolds from $(\mathrm{Pd}_{U}(\mathcal{X}),\omega^U)$ to $\mathrm{Pd}(\mathcal{X})$, we need to show that $\gamma(t)=U\delta(t)$ for all $t\in [0,1]$. Since 
    \begin{align*}
        \exp_U(t\log_U(X))=U\exp(tU^*U\log(U^*X))=U\exp(t\log(U^*X))=U(U^*X)^t
    \end{align*}
    we have that $X_U^t=U(U^*X)^t$ for all $X\in \mathrm{Pd}_U(\mathcal{X})$. So
    \begin{equation*}{A_U^{-\frac{1}{2}}}\bullet B\bullet {A_U^{-\frac{1}{2}}}=(U(U^*A)^{-\frac{1}{2}})(U^*BU^*)(U(U^*A)^{-\frac{1}{2}})=U(U^*A)^{-\frac{1}{2}}(U^*B)(U^*A)^{-\frac{1}{2}})\end{equation*}
    Hence we have 
    \begin{align*}
        \gamma(t)&=A_U^{-\frac{1}{2}}\bullet ((A_U^{-\frac{1}{2}})^{-1}\bullet B\bullet (A_U^{-\frac{1}{2}})^{-1})_U^t\bullet A_U^{-\frac{1}{2}}\\
        &=U(U^*A)^{-\frac{1}{2}}U^*U(U^*U(U^*A)^{-\frac{1}{2}}(U^*B)(U^*A)^{-\frac{1}{2}})^{t}U^*U(U^*A)^{-\frac{1}{2}}\\
        &=U(U^*A)^{\frac{1}{2}}((U^*A)^{-\frac{1}{2}}(U^*B)(U^*A)^{-\frac{1}{2}})^{t}(U^*A)^{-\frac{1}{2}}\\
        &=U\delta(t)
    \end{align*}
\end{proof}

\subsection{States, channels and measurements for S-spaces}
In this subsection, we list and occasionally construct the analogous definitions from the $J$-fidelity case to $S$-spaces, much of this is either done or alluded to in \cite{bag2024quantumuchannelssspaces}.

\begin{definition}[$U$-quantum state]
A matrix $P\in L(\mathcal{X})$ is a $U$-quantum state if it is $U$-positive and $\mathrm{Tr}(U^*P)=1$. \cite{bag2024quantumuchannelssspaces}
\end{definition}
We provide a slight re-characterization of the definitions in \cite{bag2024quantumuchannelssspaces}
\begin{definition}[$U$-positive,-completely positive map and $U$-quantum channel.]
Given $U$-spaces $(\mathcal{X},U_1)$ and $(\mathcal{Y},U_2)$, $\Phi:L(\mathcal{X})\to L(\mathcal{Y})$, 
    \begin{enumerate}
        \item $\Phi$ is $(U_1,U_2)$-positive if $\Phi$ maps $U_1$-positive elements to $U_2$-positive elements, which is equivalent to  $U_2\Phi(U_1^*P)$ being positive semi-definite for all $P\in \mathrm{Pos}(\mathcal{X})$.
        \item $\Phi$ is $(U_1,U_2)$-completely positive if given another $U$-space $(\mathcal{Z},U_3)$, the map \begin{equation*}\Phi\otimes \mathbf{1}_{L(\mathcal{Z})}:L(\mathcal{X})\otimes L(\mathcal{Z})\to L(\mathcal{Y})\otimes L(\mathcal{Z})\end{equation*} is $(U_1\otimes U_3)-(U_2\otimes U_3)$-positive, which is equivalent to \begin{equation*}U_2\Phi(U_1^*\cdot)\otimes \mathbf{1}_{L(\mathcal{Z})}\end{equation*} being positive for any $U$-space $\mathcal{Z}$.
        \item $\Phi$ is a $(U_1,U_2)$ channel if it is $(U_1,U_2)$-completely positive and $U_2\Phi(U_1^*\cdot)$ is trace-preserving. 
    \end{enumerate}
\end{definition}
We have, as in the Krein case, the following commutative diagrams. Given a $(U_1,U_2)$-positive map, the following diagram commutes
\begin{equation*}
\begin{tikzcd}[column sep=huge]
\mathrm{Pos}_{U_1}(\mathcal{X})
  \arrow[r, "\Phi"]
  \arrow[d, shift left=1, "U_1^*(\cdot)"]
&
\mathrm{Pos}_{U_2}(\mathcal{Y})
  \arrow[d, shift left=1, "U_2^*(\cdot)"] \\
\mathrm{Pos}(\mathcal{X})
  \arrow[r, "\Psi"']
  \arrow[u, shift left=1, "U_1(\cdot)"]
&
\mathrm{Pos}(\mathcal{Y}\otimes \mathcal{Z})
  \arrow[u, shift left=1, "U_2(\cdot)"]
\end{tikzcd}
\end{equation*}
where $\Psi=U_2\Phi(U_1^*\cdot)$

and for a $(U_1,U_2)$-completely positive map, the following commutes for any $U_3$
\begin{equation*}
\begin{tikzcd}[column sep=huge]
\mathrm{Pos}_{U_1\otimes U_3}(\mathcal{X}\otimes \mathcal{Z})
  \arrow[r, "\Phi\otimes \mathbf{1}_{L(\mathcal{Z})}"]
  \arrow[d, shift left=1, "U_1^*\otimes U_3^*(\cdot)"]
&
\mathrm{Pos}_{U_2\otimes U_3}(\mathcal{Y}\otimes \mathcal{Z})
  \arrow[d, shift left=1, "U_2^*\otimes U_3^*(\cdot)"] \\
\mathrm{Pos}(\mathcal{X}\otimes \mathcal{Z})
  \arrow[r, "\Psi\otimes \mathbf{1}_{L(\mathcal{Z})}"']
  \arrow[u, shift left=1, "U_1\otimes U_3(\cdot)"]
&
\mathrm{Pos}(\mathcal{Y}\otimes \mathcal{Z})
  \arrow[u, shift left=1, "U_2\otimes U_3(\cdot)"]
\end{tikzcd}
\end{equation*}

In similar style to Section \ref{Krein section} and as alluded to in \cite{bag2024quantumuchannelssspaces}, we can consider the notion of $U$-measurements. 

\begin{definition}[$U$-measurement]
    Given an alphabet $\Sigma$, a complex Euclidean space $\mathcal{X}$ and a unitary operator $U\in U(\mathcal{X})$, a $U$-measurement is a function
    \begin{equation*}\mu:\Sigma\to \mathrm{Pos}_U(\mathcal{X})\end{equation*} such that $\sum_{a\in\Sigma}\mu(a)=U$.
\end{definition}
\subsection{Motivation of S-space fidelity.}
In \ref{Krein section}, the notion of the $J$-geometric mean is useful for computing the measurement basis for which $J$-fidelity is achieved. We will demonstrate that the analogous result holds for $U$-states.

Given two $U$-states, we would like to minimize the value 
\begin{equation*}B_U(P,Q|\mu)=\sum_{a\in\Sigma}\sqrt{\langle \mu(a),P\rangle}\sqrt{\langle \mu(a),Q\rangle}\end{equation*} across all possible $U$-measurements $\mu$. Again, since this value is equal to \begin{equation*}\sum_{a\in\Sigma}\sqrt{\langle U^*\mu(a),U^*P\rangle}\sqrt{\langle U^*\mu(a),U^*Q\rangle}\end{equation*} we have by the regular quantum theory 
that the minimization of this over all $\mu$ is given by 
\begin{equation*}F(U^*P,U^*Q)=\mathrm{Tr}(\sqrt{U^*Q}U^*P\sqrt{U^*Q})=\|\sqrt{U^*P}\sqrt{U^*Q}\|_1\end{equation*}
\begin{definition}
    Given two $U$-quantum states, $P$ and $Q$, the $U$-fidelity of $P$ and $Q$ to be 
    $F_U(P,Q)=F(U^*P,U^*Q)$
\end{definition}
We now verify for completeness that this corresponds to taking a $U$-measurement in the '$U$-geometric' mean of $P$ and $Q^{-1}$. Again, we are just following the analysis of Chapter 3 of \cite{Watrous_2018} and the proof contained in \cite{wilde2015lecture16}. Assuming that $P$ and $Q$ are invertible, we have that the $U$-geometric mean of $P$ and $Q^{-1}$ is $U$-positive definite. Multiplying by $U^*$, we obtain a positive definite operator
    \begin{equation*}
        R=\sum_{i}\lambda_ix_ix_i^*
    \end{equation*}
    We then show that the optimal measurement is with respect to $x_ix_i^*$. We note that $RU^*QR=U^*P$
     since 
     \begin{align*}
         UR&=Q_U^{-\frac{1}{2}}\bullet\sqrt[U]{Q_U^{\frac{1}{2}}\bullet P\bullet Q_U^{\frac{1}{2}}}\bullet Q_U^{-\frac{1}{2}}\\
         &= U(U^*Q)^{-\frac{1}{2}}(\sqrt{U^*U\sqrt{U^*Q}U^*PU^*U(\sqrt{U^*Q})})U^*U(U^*Q)^{-\frac{1}{2}}\\
         &=U(U^*Q)^{-\frac{1}{2}}(\sqrt{\sqrt{U^*Q}U^*P(\sqrt{U^*Q})})(U^*Q)^{-\frac{1}{2}}
     \end{align*}
     Hence multiplying by $U^*$, we obtain 
     \begin{equation*}R=(U^*Q)^{-\frac{1}{2}}(\sqrt{\sqrt{U^*Q}U^*P(\sqrt{U^*Q})})(U^*Q)^{-\frac{1}{2}}\end{equation*}
     We then have 
     \begin{align*}
         RU^*QR=(U^*Q)^{-\frac{1}{2}}(\sqrt{U^*Q}U^*P(\sqrt{U^*Q})(\sqrt{U^*Q})^{-\frac{1}{2}}=U^*P   
     \end{align*}
     We then have that 
    \begin{align*}
        &\sum_{i}\sqrt{\mathrm{Tr}((U^*x_ix_i^*)^*P)\mathrm{Tr}((U^*x_ix_i^*)^*Q)}\\=& \sum_{i}\sqrt{\mathrm{Tr}(x_ix_i^*U^*P)\mathrm{Tr}((x_ix_i^* U^*Q)}\\
        =& \sum_i \sqrt{(x_i^*RU^*QRx_i)(x_i^*U^*Qx_i)}\\
        =& \sum_i \sqrt{(x_i^*\lambda_i U^*Q\lambda_ix_i)(x_i^*U^*Qx_i)}\\
        =& \sum_i \lambda_ix_i^*U^*Qx_i\\
        =&\mathrm{Tr}\left(\sum_{i}\lambda_ix_ix_i^*U^*Q\right)\\=&\mathrm{Tr}(RU^*Q)\\
        =&\mathrm{Tr}\left(\sqrt{\sqrt{U^*Q}U^*P\sqrt{U^*Q}}\right)\\
        =&F(U^*P,U^*Q)\\
        =&F_U(P,Q)
    \end{align*}
    
    If $\mu:\Sigma\to \mathrm{Pos}_U(\mathcal{X})$ is a $U$-POVM, then $U^*\mu$ is a POVM and hence by the regular quantum theory \cite{Watrous_2018}, we have 
    \begin{align*}
        F(U^*P,U^*Q)\leq \sum_{a\in \Sigma}\sqrt{\langle U^*\mu(a), U^*P\rangle}\sqrt{\langle U^*\mu(a),JQ\rangle}=\sum_{a\in \Sigma}\sqrt{\langle \mu(a), P\rangle}\sqrt{\langle \mu(a),Q\rangle}
    \end{align*}
    Hence, in the case $P$ and $Q$ are invertible, we have the $U$-fidelity is defined as 
    \begin{equation*}F_U(P,Q)=\mathrm{min}_{\mu}\sum_{a}\sqrt{\langle \mu(a),P\rangle}\sqrt{\langle \mu(a), Q\rangle}=F(U^*P, U^*Q)\end{equation*}
    In the case that $P$ and $Q$ are not invertible, we just take the appropriate projections as in the $J$-case.

\subsection{Properties of S-space fidelity}
We list the immediate corollaries of the definition of $U$-fidelity. The following is largely a repeat of \ref{conseqeunces of definiton}, which in turn is all adapting the results contained in Chapter 3 of \cite{Watrous_2018} to the slightly different $J$-space case. Where a proof is not given it is immediate from either \ref{conseqeunces of definiton} or \cite{Watrous_2018}.
\begin{proposition}
   We have the following properties hold given $S$-space $(\mathcal{X},U)$;
    \begin{itemize}
    
        \item The $U$-fidelity function is continuous at $P,Q \in \mathrm{Pos}_U(\mathcal{X})$.
        \item $F_U(\lambda P,Q)=\sqrt{\lambda}F_U(P,Q)=F(P,\lambda Q)$ for all real $\lambda\geq 0$.
        \item $F_U(P,Q)=F(\Pi_{(\mathrm{im}U^*Q)}U^*P\Pi_{(\mathrm{im}U^*Q)}, U^*Q)=F(P,\Pi_{(\mathrm{im}(U^*P))}U^*Q\Pi_{(\mathrm{im}U^*P)})$.
        \item $F_U(P,Q)\geq 0$ with equality iff $PQ=0$
        \item $F_U(P,Q)^2\leq \mathrm{Tr}(U^*P)\mathrm{Tr}(U^*Q)$ with equality iff $P$ and $Q$ are linearly dependent.
    \end{itemize}
\end{proposition}
\begin{proof}
    See Proposition 3.12 of \cite{Watrous_2018}.
\end{proof}

\begin{proposition}
    Let $\mathcal{X}$ be a complex Euclidean space and let $v\in \mathcal{X}$ be a vector and let $P\in \mathrm{Pos}_U(\mathcal{X})$. We have that 
    \begin{align*}
        F_U(P, U^*vv^*)=\sqrt{v^*U^*Pv}
    \end{align*}
    and in particular we have that $F(U^*uu^*, U^*vv^*)=|\langle u,v\rangle |$
\end{proposition}
\begin{proof}
    See Proposition 3.13 of \cite{Watrous_2018}.
\end{proof}

\begin{proposition}
    Let $P,Q \in\mathrm{Pos}_U(\mathcal{X})$, then 
    \begin{align*}
        F_U(P,QU^*PU^*Q)=\langle P,Q \rangle
    \end{align*}
\end{proposition}
 \begin{proof}We follow the proof of Proposition 3.14 \cite{Watrous_2018} with the appropriate changes.
     \begin{align*}
     F_U(P,QPQ)&=F(U^*P,U^*QU^*PU^*Q)\\
     =&\mathrm{Tr}(\sqrt{\sqrt{U^*P}U^*QU^*PU^*Q\sqrt{U^*P}})\\
     =&\mathrm{Tr}(\sqrt{U^*P}U^*Q\sqrt{U^*P})\\
     =&\mathrm{Tr}(U^*PU^*Q)\\
     =&\langle U^*P, U^*Q\rangle\\
     \end{align*}
 \end{proof}
\begin{lemma}[Winter's gentle measurement lemma for $U$-fidelity]
Let $\mathcal{X}$ be a complex Euclidean space with $U$-inner-product and let $\rho$ be a $U$-quantum state and let $P\in \mathrm{Pos}_U(\mathcal{X})$ be a positive semi-definite operator satisfying $U^*P\leq \mathbf{1}_{\mathcal{X}}$ and $\langle P,\rho\rangle>0$. Then
\begin{equation*}F_U\left(\rho, \frac{\sqrt{U^*P}U^*\rho U^*\sqrt{U^*P}}{\langle U^*P, U^*\rho\rangle}\right)\geq \langle P,\rho\rangle\end{equation*}
\end{lemma}
\begin{proof}
    We have (again following the proof structure in \cite[p.~143]{Watrous_2018} that
    \begin{align*}
        F_U\left(\rho,\frac{U\sqrt{U^*P}U^*\rho \sqrt{U^*P}}{\langle U^*P,U^*\rho\rangle}\right)=\frac{1}{\sqrt{\langle U^*P,U^*\rho\rangle}}F(U^*\rho, \sqrt{U^*P}U^*\rho \sqrt{U^*P})=\frac{\langle \sqrt{U^*P},U^*\rho\rangle}{\sqrt{\langle U^*P, U^*\rho\rangle}}
    \end{align*}
    and from the assumption on $U^*P$ we have the result follows. 
\end{proof}

\subsection{Characterizing S-space fidelity}
Let $\mathcal{X}$ be a complex Euclidean space with unitary inner-product $[\cdot,\cdot]_U$, and let $P, Q \in \mathrm{Pos}_U(\mathcal{X})$. We have
\begin{align*}
    F_U(P,Q)
    = \max \left\{ \left| \mathrm{Tr}(X) \right| : X \in L(\mathcal{X}), \;
    \begin{bmatrix}
        U^*P & X \\
        X^{*} & U^*Q
    \end{bmatrix}
     \in \mathrm{Pos}(\mathcal{X}\oplus\mathcal{X})
    \right\}.
\end{align*}

In the exact same way as seen in \cite[p.~147]{Watrous_2018}, we obtain the semi-definite program view of fidelity in that the $U$-fidelity $F_U(P,Q)$ gives the optimal value of the primal problem given by maximizing 
 \begin{align*}
     \frac{1}{2}\mathrm{Tr}(X)+\frac{1}{2}\mathrm{Tr}(X^*)
 \end{align*}
 subject to 
 \begin{align*}
     \begin{bmatrix}
        U^*P & X \\
        X^{*} & U^*Q
    \end{bmatrix}\geq 0
 \end{align*} for $X\in L(\mathcal{X})$.
We may also obtain the $U$-formulation of Alberti's theorem
\begin{theorem}[$U$ version of Alberti's theorem]
    For $P,Q\in \mathrm{Pos}_U(\mathcal{X})$ we have 
    \begin{align*}
        F_U(P,Q)^2=\inf\{ [P, Y]_U[Q,Y]_U^{-1}: Y\in \mathrm{Pd}_U(\mathcal{X})]\}
    \end{align*}
\end{theorem}
\begin{proof}
    See \cite[pp.~148--149]{Watrous_2018}.
\end{proof}

We can also consider $U$-partial traces as in the $J$-case. If $P\in \mathrm{Pos}_U(\mathcal{X})$, then there exists a complex Euclidean space $\mathcal{Y}$ and a vector $u\in \mathcal{X}\otimes \mathcal{Y}$ such that \begin{equation*}\mathrm{Tr}_{\mathcal{Y}}((U\otimes I)(U^*\otimes I)uu^*))=U^*P\end{equation*} if and only if $\dim(\mathcal{Y})\geq \mathrm{rank}(P)$. We call the function $\mathrm{Tr}_{\mathcal{Y}}(U\cdot)$ the $U$-partial trace and similarly, we consider the $U$ trace to be $\mathrm{Tr}(U\cdot)$ (remembering $U\otimes I$ is still a unitary) and we will denote the $U$-partial trace $\mathrm{Tr}^U_{\mathcal{Y}}$ and the $U$-trace $\mathrm{Tr}^U$. 

\subsection{U-analog of Uhlmann's theorem}
Again, we are able to adapt Uhlmann's theorem to the $U$-case. 
\begin{lemma}
    Let $A,B\in L(\mathcal{Y},\mathcal{X})$ be operators for complex Euclidean spaces $\mathcal{X}$ and $\mathcal{Y}$ and endow $\mathcal{X}$ with unitary $U$. We have 
    \begin{align*}
        F_U(UAA^*,UBB^*)=\|A^*B\|_1
    \end{align*}
\end{lemma}
\begin{proof} We have that 
    \begin{equation*}
    F_U(UAA^*,UBB^*)=F(AA^*,BB^*)
    \end{equation*}
    from which we may then apply the proof of Lemmma 3.21 of \cite{Watrous_2018} to the operators $AA^*$ and $BB^*$.
\end{proof}

\begin{theorem}
    Let $\mathcal{X}$ and $\mathcal{Y}$ be complex Euclidean spaces and let $U_1$ be unitary matrix for $\mathcal{X}$ and $U_2$ be a unitary matrix for $\mathcal{Y}$ and let $P,Q\in \mathrm{Pos}_U(\mathcal{X})$ be U-positive semi-definite operators having rank at most $\mathrm{dim}(\mathcal{Y})$ and let $u\in \mathcal{X}\otimes {Y}$ satisfy $\mathrm{Tr}^{U_1\otimes U_2}_{\mathcal{Y}}((U_1^*\otimes U_2^*uu^*))=U_1^*P$. It holds that 
    \begin{align*}
        F_U(P,Q)=\max\{|\langle u, v\rangle|:v\in \mathcal{X}\otimes \mathcal{Y}, \mathrm{Tr}_{\mathcal{Y}}^{U_1\otimes U_2}((U_1^*\otimes U_2^*)vv^*) \}
    \end{align*}
\end{theorem}
\begin{proof} We adapt the proof contained in \cite{Watrous_2018} with the appropriate adjustments. 
    Let $A\in L(\mathcal{Y},\mathcal{X})$ be the operator such that $u=\mathrm{vec}(A)$ and let $w\in \mathcal{X}\otimes \mathcal{Y}$ be the vector satisfying $U^*_1Q=\mathrm{Tr}_{\mathcal{Y}}^{U_1\otimes U_2}((U_1\otimes U_2)ww^*)$ and let $B\in L(\mathcal{Y},\mathcal{X})$ be the operator for which $w=\mathrm{vec}(B)$. We then have 
    \begin{align*}
        &\mathrm{max}\{|\langle u,v\rangle|: v\in \mathcal{X}\otimes \mathcal{Y}, \mathrm{Tr}_{\mathcal{Y}}^{U_1\otimes U_2}((U_1^*\otimes U_2^*)vv^*)=U_1^*Q\}\\
        =& \max\{|\langle u, (\mathbf{1}_{\mathcal{X}}\otimes S)w\rangle| :S\in U(\mathcal{Y})\}\\
        =&\max{|\langle A,BS^T\rangle|:S\in U(\mathcal{Y})}\}\\
        =& \max\{|\langle \overline{S}, A^*B\rangle|:S\in U(\mathcal{Y})\}\\
        =&\|A^*B\|_1=F_{U}(U_1AA^*,U_1BB^*)=F_{U}(P,Q)
    \end{align*}
\end{proof}
\begin{corollary}
Let $u, v \in \mathcal{X} \otimes \mathcal{Y}$ be vectors for complex Euclidean spaces $\mathcal{X}$ and $\mathcal{Y}$ equipped with unitary matrices $U_1$ and $U_2$, so that $\mathcal{X} \otimes \mathcal{Y}$ is equipped with unitary matrix $U_1 \otimes U_2$. We have
\begin{align*}
F_{U_1}\Big(
\mathrm{Tr}_{\mathcal{Y}}^{U_1 \otimes U_2}\big((U^*_1 \otimes U^*_2)(uu^*)\big),
\mathrm{Tr}_{\mathcal{Y}}^{U_1 \otimes U_2}\big((U^*_1 \otimes U^*_2)(vv^*)\big)
\Big).
\end{align*}
\end{corollary}

\begin{proof} Again, we follow the method in \cite{Watrous_2018}.
Let $A,B \in L(\mathcal{Y},\mathcal{X})$ be such that $u=\mathrm{vec}(A)$ and $v=\mathrm{vec}(B)$.
\begin{align*}
&F_{U_1}\Big(
\mathrm{Tr}_{\mathcal{Y}}^{U_1 \otimes U_2}\big((U^*_1 \otimes U^*_2)(uu^*)\big),
\mathrm{Tr}_{\mathcal{Y}}^{U_1 \otimes U_2}\big((U_1^* \otimes U_2^*)(vv^*)\big)
\Big) \\
&= F_{U_1}(U_1A A^*,\, U_1 B B^*) \\
&= \|A^* B\|_1 \\
&= \|(A^* B)^T\|_1 \\
&= \left\|\mathrm{Tr}_{\mathcal{X}}^{U_1 \otimes U_2}\big((U_1^* \otimes U^*_2)(v u^*)\big)\right\|_1.
\end{align*}
\end{proof}
\section{Further questions:}
The cases of $J$-fidelity and $U$-fidelity are a little boring. We are performing the original analysis of quantum fidelity but in various stages we ``twist" or ``untwist" the cone of positive (semi)-definite matrices. A more interesting question is what if one relaxes the conditions on the map through which our inner-product is defined? Do we still obtain a natural notion of fidelity, which coincides both with the motivation of minimizing certain Bhattacharryya coefficients through measuring in the eigenbasis of some notion of the geometric mean?
\section*{Acknowledgements}
The author would like to thank Doctor Javier Alejandro Ch\'avez Dom\'inguez for continuing support and supervision on this project, and for the introduction to the world of quantum information. The author would also like to thank Doctor Keri Kornelson and Doctor Alex Grigo for early useful chats about possible directions and suggestions for topics to consider. The author also used ChatGPT 5.5 to search for references of basic results and to aid with the literature review stage, as well as fixing some LaTeX issues. 
\nocite{bhatia2009positive}
\nocite{bhatia2013matrix}
\nocite{johnston2021advanced}
\nocite{*}
\bibliographystyle{amsalpha}
\bibliography{Bibliography}
\end{document}